\documentclass[%
 onecolumn,
 amsmath,amssymb,
 aps, 11pt
]{revtex4-2}

\usepackage[utf8]{inputenc}
\usepackage{graphicx}
\usepackage{dcolumn}
\usepackage{bm}
\usepackage{tikz,tikz-3dplot}
\usepackage{mathrsfs}
\usepackage{enumitem}
\usepackage{amsthm}
\usepackage{dsfont}
\usepackage{mathtools}

\DeclareMathOperator{\trusts}{\raisebox{2pt}{\rotatebox[origin=c]{180}{$\leadsto$}}}
\DeclareRobustCommand{\id}{\mathds{1}}

\NewDocumentCommand{\rest}{sm}{%
	\vcenter{\hbox{\scalebox{1.1}[2.0]{$\upharpoonright$}}}%
	\IfBooleanTF{#1}{%
		_{\!#2}%
	}{%
		_{\mathrlap{\!#2}}\,%
	}%
}

\newtheorem{definition}{Definition}[section]
\newtheorem{theorem}{Theorem}[section]

\begin{document}

\title{The Contextual Modal Logic of a Wigner's Friend Generalization}

\author{Felipe D. Alves}
 \altaffiliation{felipe.dilho.alves@usp.br}
\author{João C. A. Barata}%
 \email{jbarata@if.usp.br}
\affiliation{%
 Departamento de Física-Matemática,\\
Instituto de Física - Universidade de São Paulo, Brazil
}

\date{\today}

\begin{abstract}
Quantum mechanics has been subject to logical scrutiny since its inception. The behavior of quantum systems, which are fundamentally dissimilar from classical systems, often appears to point to a logical inconsistency in quantum mechanics, allegedly leading to contradictions in the prediction of experimental measurements--though such contradictions have never materialized. A recent example of this type of inquiry into the logical well-posedness of quantum mechanics is the Frauchiger-Renner Gedankenexperiment, which purports to demonstrate that quantum mechanics is logically inconsistent. In this article, we show that by considering the property of contextuality in quantum systems--as predicted by the Kochen-Specker theorem--the supposed contradiction proposed by Frauchiger and Renner becomes logically inaccessible.
\end{abstract}
\maketitle
\section{Introduction}

One of the first proposed problems with the logical consistency of quantum mechanics was regarding Schrödinger's cat Gedankenexperiment \cite{Cat}, which amplified quantum phenomena to the macroscopic world. Of course, Schrödinger himself didn't find that a contradiction, but a way to characterize the nature of quantum systems. Since the possibility of having a system such as a cat be in a superposition of mutually exclusive states went directly against classical logic, but this didn't make quantum mechanics inconsistent with itself, only with a classical world.

A second such theoretical consideration, the Wigner's friend Gedankenexperiment \cite{Friend}, considered a observer as itself a quantum mechanical system, and the equivalence between observers as being equally privileged. In a sense, Wigner's Gedankenexperiment could be seen as a dual to Schrödinger's, whereas Schrödinger sought to contrast quantum descriptions to classical ones by bringing wave-function collapse to be a element of reality of the macroscopic world; Wigner's Gedankenexperiment instead calls into question whether wave-function collapse can be regarded as an objective physical process at all. Since the collapse his friend observes is in no way more fundamental than the collapse Wigner himself observes of the laboratory as a whole when the information of his friend's measurement result is assessed.

Nevertheless, the considerations invoked by Wigner's Gedankenexperiment also don't lead to any inconsistencies of quantum mechanics with itself. A more recent thought experiment proposed by D. Frauchinger and R. Renner \cite{FrauRen}, and henceforth referenced as the Frauchinger-Renner Gedankenexperiment, seems to in fact arrive at a fundamental inconsistency of quantum mechanics, where a measurement result that is calculated to have a zero probability of occurring actually occurs $\frac{1}{12}$ of the times. The fundamental parts of the argument are reproduced here for completeness, and this apparent inconsistency is shown to actually be theoretically unreachable, unless the underlying logical structure of quantum mechanics is taken to be non-contextual, which we know not to  be the case because of the Kochen-Specker theorem \cite{KochenSpecker}.

Contextuality in this situation involves the structure of the von Neumann algebras of the observables of the theory, more specifically the set of commutative von Neumann subalgebras of the physical system in question, this observation combined with the general treatment of quantum field theory though algebraic methods given by algebraic quantum field theory \cite{IntAQFT,Haag,Araki2009} allows these questions and constructions to be brought to QFT, where the complications involved in such a generalization will be discussed.

\section{Generalized Wigner's Friend Gedankenexperiment}

The Frauchinger-Renner Gedankenexperiment \cite{FrauRen} can be described as a generalization of the original Wigner's friend Gedankenexperiment. If Wigner's thought experiment can be described by a quantum system $S$ and a agent $F$, himself also described as a quantum system, that will perform a measurement on $S$ and both are within a laboratory $L$ that isolates $S$ and $F$ from the exterior until the planned measurement that $F$ realizes on $S$ occurs, after which another agent $W$ that is exterior to the lab $L$ measures the combined quantum system $ S \otimes F $. Henceforth, assuming that the wave-function collapse is a element of reality, the question is posed if the collapse occurred when $F$ measures $S$ or when $W$ measures $S \otimes F$.

\begin{figure}[h!]
	\centering
	
	\hspace*{1.7cm}
	\scalebox{1.5}{
		\begin{tikzpicture}[->]
			(1,-1.2) edge[bend right] node [left] {} (-1,-1.2);
			\draw[black, very thick] (0.8,-1.7) rectangle (2.3,-0.7);
			\node at (2.55,-1.75) (bL) {$L$};
			\node at (1.85,-1.2) (bF) {$F$};
			\node at (1.27,-1.2) (bd) {$S$};
			\node at (3.6,-1.2) (bW) {$W$};
			\node at (2.28,-0.9) (ExbW) {};

			\path[every node/.style={font=\sffamily\small}, thick]
			(bW) edge[bend right] node [left] {} (ExbW)
			(1.9,-1) edge[bend right] node [left] {} (1.4,-1);
			
	\end{tikzpicture} } \caption{Representation of Wigner's friend Gedankenexperiment, where the quantum system $S$ is measured by $F$ in the laboratory $L$ and $W$ measures the entirety of $L$. }
\end{figure}
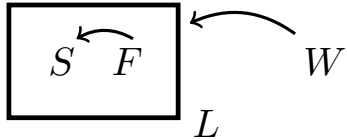

This question is meaningless in the usual Copenhagen interpretation of quantum mechanics, where wave-function collapse isn't a physical process, but a theoretical way to update probability amplitudes with the increase of information about the system obtained by the measurement. As is realized in \cite{HeppM} and criticized in the corresponding discussion in \cite{BellM}.

The Frauchinger-Renner thought experiment can then be described by a duplication of the ensemble of Wigner's friend Gedankenexperiment, with a more complicated protocol of the  measurements that are to be realized and their prescribed order. We then have two laboratories $L_1$ and $L_2$, where in each we can find a two-level quantum system $S_i$ and a respective agent $F_i$ that will realize a measurement on system $S_i$, $i=1,2$. We also have two outside observers $W_1$ and $W_2$. 

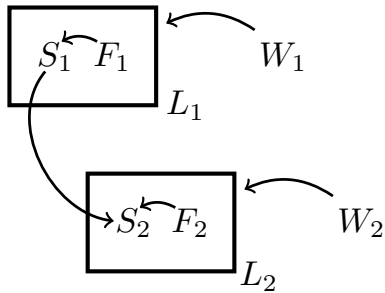
\begin{figure}[h!]
	\centering
	\scalebox{1.3}{
		\begin{tikzpicture}[->]
			\draw[black, very thick] (0,0) rectangle (1.5,1);
			\node at (1.8,0) (aL) {$L_1$};
			\draw[black, very thick] (0.8,-1.7) rectangle (2.3,-0.7);
			\node at (2.55,-1.75) (bL) {$L_2$};
			\node at (1.05,0.5) (aF) {$F_1$};
			\node at (0.47,0.5) (ad) {$S_1$};
			\node at (1.85,-1.2) (bF) {$F_2$};
			\node at (1.27,-1.2) (bd) {$S_2$};
			\node at (2.8,0.5) (aW) {$W_1$};
			\node at (3.6,-1.2) (bW) {$W_2$};
			\node at (1.48,0.8) (ExaW) {};
			\node at (2.28,-0.9) (ExbW) {};
			\node at (0.41,0.58) (ad1) {};
			\node at (1.02,0.58) (aF1) {};
			\node at (1.21,-1.12) (bd2) {};
			\node at (1.82,-1.12) (bF2) {};
			\node at (1.19,-1.2) (bd3) {};
			\node at (0.47,0.47) (ad3) {};
			
			\path[every node/.style={font=\sffamily\small}, thick]
			(ad3) edge [out=230, in=173] (bd3)
			(aW) edge[bend right] node [left] {} (ExaW)
			(bW) edge[bend right] node [left] {} (ExbW)
			(aF1) edge[bend right] node [left] {} (ad1)
			(bF2) edge[bend right] node [left] {} (bd2);
	\end{tikzpicture} }
	\caption{Representation of the Frauchinger-Renner Gedankenexperiment as a variation in the Wigner's friend Gedankenexperiment.}
\end{figure}

It is considered, just as in the Wigner's friend Gedankenexperiment, that each experimenter $F_1$ and $F_2$ can have their states measured to determine the measured outcome of the two level system that they perceived, as a simplifying assumption we shall consider this to be representable also with a two-level quantum system. One can, of course, criticize this simplification as non-representative of the complexity that is a rational agent, in which case the same experiment can be carried out mechanically with this second two-level quantum system being used as a replacement for the state of the experimenter memory, where the measurement result is held.

We can then briefly describe the experimental protocol that every measuring agent knows before the experiment can even begin. To arrive at a instance where the contradiction occurs, the experiment is repeated several times, and hence each step of the experiment happens at a given, predefined, time $t^{[n]}$, where the $[n]$ superscript holds the number of repetitions the experiment has already undertaken, with the numerical time being zeroed after each trial. This is completely analogous to the setting found in another Gedankenexperiment that is commonly referred as Hardy's paradox \cite{Hardy}, where a couple of Mach-Zehnder interferometers are interposed, though one is sent a electron beam whereas a positron beam is sent in the other, these interferometers have one point where beam paths from different interferometers intersect and particle annihilation can take place. 

The supposed paradox happens because, when imposing locality and classical realism conditions, there is an instance in which at least one of the detectors placed at a destructive interference point of the beams actually detects a particle. Then as the only influence placed on either interferometer is the possible annihilation. This means that a annihilation is influencing the measurement outcomes, but in a situation in which only one electron and positron are present in each interferometer beam.

An annihilation means that no detection should occur in either detector since the particles were converted into photons, this establishes a paradox of sorts when trying to force classical notions on this quantum system. Since the Frauchinger-Renner Gedankenexperiment also involves this dependence on the configuration of the states in the system, it can also be seen as a generalization of Hardy's paradox \cite{Hardy}.

Going back to the description of the experimental protocol, we consider a particular basis $\{ \varphi , \psi \} \subset \mathscr{H}_S$ to represent the two-level quantum systems $S_i$ of the experiment and another basis $\{ \xi , \zeta \} \subset \mathscr{H}_F$ to represent the agents $F_i$, that have been assumed to be describable as two-level systems as well. The Gedankenexperiment is  outlined by the following protocol:

\begin{enumerate}[label=Step $t^{[n]}$ \text{=} $\arabic*:$, start=0, leftmargin=*, labelsep=1em]
	\item Agent $F_1$ prepares a quantum system $S_1$ to be in a state $\sqrt{\frac{1}{3}} \varphi + \sqrt{\frac{2}{3}}\psi$, and subsequently measures the system $S_1$ in the orthogonal basis $\{ \varphi , \psi \} \subset \mathscr{H}_S$ and if the result of this measurement is $\varphi$ then $F_1$ sends a quantum system $S_2$ prepared in state $\varphi \in \mathscr{H}_S$ to lab $L_2$, otherwise, if the result of the measurement	is $\psi$, then $F_1$ sends the quantum system $S_2$ prepared in the state $\sqrt{\frac{1}{2}}(\varphi+\psi)$ to lab $L_2$.
	
	\item Agent $F_2$ measures $S_2$ with respect to the basis $\{ \varphi , \psi \}$  and preserves the resulting state as a recording.
	
	\item Agent $W_1$ measures the entirety of the lab $L_1$ in a orthogonal basis given by $\left\{\sqrt{\frac{1}{2}} (\varphi \otimes \xi + \psi \otimes \zeta), \sqrt{\frac{1}{2}} (\varphi \otimes \xi - \psi \otimes \zeta), \sqrt{\frac{1}{2}} (\psi \otimes \xi + \varphi \otimes \zeta), \sqrt{\frac{1}{2}} (\psi \otimes \xi - \varphi \otimes \zeta) \right\}$,  more specifically with the projector $P_{\sqrt{\frac{1}{2}} (\varphi \otimes \xi - \psi \otimes \zeta)} \equiv P_\chi$ and if he obtains a non-null outcome he then announces the information ``$P_\chi L_1 \not=0$" publicly to every other agent in the experiment, otherwise he announces ``$P_\chi L_1 =0$".
	
	\item Agent $W_2$ measures the entirety of the lab $L_2$ in a equivalent orthogonal basis to that of $W_1$, also though the projector $P_\chi$ and if he obtains a non-null outcome he then announces the information ``$P_\chi L_2 \not= 0$" publicly to every other agent in the experiment, otherwise he announces ``$P_\chi L_2 =0$".
	
	\item If ``$P_\chi L_1 \not=0$" and ``$P_\chi L_2 \not= 0$" then the experiment is halted, otherwise the index of $t^{[n]}$ is increased to $t^{[n+1]}$ and the procedure is repeated beginning from the step time $0$. This is repeated until the experiment halts.
\end{enumerate}

In the literature about the subject it is recurrent the use of naïve, informal logical inferences to arrive at the contradiction originally arrived by Frauchinger and Renner. this has been known to be improper since before the work of Kochen and Specker \cite{KochenSpecker} in the 60's, which nevertheless crystallized the problem with their construction involving 117 Hilbert space vectors.

A good reference for both the naïve approach to the problem, as well as a more structured and formal approach is present in \cite{Nur}. We will use formal symbolic logic as the language to make precise the deductions each agent in the Gedankenexperiment does, bringing these logical deductions closer to mechanical computations and thereby greatly reducing the possibility of confusions or of an agent acquiring more information than is available to them.

\section{Formal Deductions}

As a brief exposition of the logical notation which I will employ in the following logical deductions, I  will define what is know as a Kripke Structure \cite{ModalL} (or Kripke semantics) which will codify the meaning of each measurement operation to each actor. For the syntax we will be using epistemic modal logic, which is a generalization of classical modal logic, itself being a extension of propositional logic. The epistemic part comes from the intention of creating a logic pertinent to information and its receptacles.

In modal logic two other types of unary operators are considered on the propositions $\square q$ (``it is necessary $q$") and $\lozenge p$ (``it is possible $p$"). These are dual by $\square q = \neg \lozenge (\neg q)$, $p$ is necessary is equal to not possibly not $p$, i.e., the necessity of $p$ is the same as the negation of the impossibility of $p$; and $\lozenge q = \neg \square(\neg q)$, $q$ is possible is equal to not necessarily not $q$, i.e.,  the possibility of $q$ is the same as the negation of the unnecessariness of $q$; 

Epistemic modal logic is generalized from modal logic in such a way that $\square q$ are replaced by $\mathscr{K}q$ that can be read as ‘‘it is epistemically necessary that $q$"$\,$or ‘‘it is known that $q$", the $\lozenge$ is not generalized or given any importance, since it does not have a clear interpretation as to what it would represent in terms of knowledge.

In a situation where we have several agents, we now have many unary operators of the type $\mathscr{K}_\text{A}, \mathscr{K}_\text{B}, \mathscr{K}_\text{C }, ...$ each referring to a agent A, B, C, $...$ so that, having introduced the syntax, we can now go on to deal with the Kripke semantic.

\begin{definition}[Kripke structure]
	A Kripke structure $M$ for $n$ \emph{agents} over a set of propositions $\Phi$ is a $(n+2)$-tuple $(\Sigma,\pi,\mathbf{R}_1 ,...,\mathbf{R}_n )$, where $\Sigma$ is a nonempty set of states, $\pi: \Sigma \times \Phi \rightarrow \{0,1\}$ is an interpretation which defines the truth value of a proposition $\phi \in \Phi$ in a state $s \in \Sigma$ and $\mathbf{R}_i$ is a binary accessibility relation on the set of states $\Sigma$ according to \emph{agent}  $i$. 
	
	We then define the evaluation $\Vdash$ relation recursively, for $s \in \Sigma$ and $\phi \in \Phi$, $s \Vdash \phi$ means:
	
	\begin{enumerate}[label=\roman*.]
		\item Atomic propositions:
		$$ p \text{ is atomic and } s \Vdash p \iff \pi(s,p)=1.$$
		\item Boolean connectives:
		$$s \Vdash \neg \phi \iff s \not\Vdash \phi,$$
		$$s \Vdash \phi_1 \wedge \phi_2 \iff  (s \Vdash \phi_1) \wedge (s \Vdash \phi_2).$$
		\item Epistemic modality:
		$$s \Vdash \mathscr{K}_i \phi \iff \forall t \in \Sigma, (s,t) \in \mathbf{R}_i \Rightarrow t \Vdash \phi. $$
	\end{enumerate}
\end{definition}

Basically, the binary accessibility relations $\mathbf{R}_i$ allow a measurement information $\mathscr{K}_i \phi$ of agent $i$ that is present in a given state $s$ of our system, to be encapsulated within a internal state $t$ of agent $i$ where $\phi$ can then be taken to be true.

We then list a number of deduction rules for rational agents:

\begin{definition}[Distribution axiom]\label{distr}
	This axiom establishes that if a agent $i$ knows $\phi$ and that $\phi \Rightarrow \psi$ then $i$ knows $\psi$, or in symbols: 
	\begin{equation}
		s \Vdash (\mathscr{K}_i\phi) \wedge (\mathscr{K}_i[\phi \Rightarrow \psi]) \iff  s \Vdash \mathscr{K}_i[\phi \wedge (\phi \Rightarrow \psi)] \implies  s \Vdash \mathscr{K}_i\psi.
	\end{equation}
	This is a realization of \emph{modus ponens} as a rule of inference.
\end{definition}

\begin{definition}[Knowledge generalization axiom]
	This axiom says that if a proposition $\phi$ is a tautology, then this proposition is available to all agents, and therefore any agent $i$ knows $\phi$, that is:
	\begin{equation}
		\Vdash \phi \implies \,\,\,\,\,\, \Vdash \mathscr{K}_i\phi,\: \forall i.
	\end{equation}
	where we have used the empty state notation $\,\,\,\Vdash \phi$ to indicate that $\phi$ is true in all states in $\Sigma$.
\end{definition}

\begin{definition}[Positive introspection axiom]\label{posit}
	Agents must intuitively be able to introspect about their own knowledge,  that is if a agent $i$ knows $\phi$ then he must also know that he knows $\phi$, so we add this property as a axiom:
	\begin{equation}
		s\Vdash \mathscr{K}_i\phi \implies s \Vdash \mathscr{K}_i \mathscr{K}_i \phi.
	\end{equation}
	
	This is equivalent by contrapositive to the statement that if it is false that agent $i$ knows that he knows $\phi$ then $i$ doesn't know $\phi$:
	\begin{equation}
		s \Vdash \neg\mathscr{K}_i\mathscr{K}_i\phi \implies s \Vdash \neg\mathscr{K}_i \phi.
	\end{equation}
\end{definition}

\begin{definition}[Negative introspection axiom]
	Agents must intuitively also be able to introspect about what they do not know, this makes sense in a logical system about knowledge where there is no confusion about what is known and about that which is unknown, that is, knowledge is always known or not known, therefore if a agent $i$ doesn't know $\phi$ then he knows that he doesn't know $\phi$, or in symbols:
	\begin{equation}
		s \Vdash \neg\mathscr{K}_i\phi \implies  s \Vdash \mathscr{K}_i\neg\mathscr{K}_i\phi.
	\end{equation}
	
	Again by contrapositive this statement is equivalent to saying that if it is false that agent $i$ knows that he doesn't know $\phi$ then he knows $\phi$, otherwise he would know that he doesn't know $\phi$, in symbols:
	\begin{equation}
		s \Vdash \neg\mathscr{K}_i \neg \mathscr{K}_i  \phi \implies   s \Vdash \mathscr{K}_i  \phi.
	\end{equation}
\end{definition}

Until this last axiom, the implications have been trivial enough that they were valid irrespective of whether we were talking about a classical system or a quantum system, but the next (which would also be the last) breaks this pattern. Originally, in usual epistemic modal logic the \emph{so called} knowledge axiom, $s \Vdash \mathscr{K}_i\phi \iff s \Vdash \phi$ would allow a encapsulated individual knowledge to escape this confinement. The reasoning behind it could be expressed as: ``Since we are talking about knowledge, it is reasonable that if an agent $i$ knows $\phi$ then $\phi$ must be necessarily true, otherwise $i$ wouldn't \textbf{know} $\phi$". 

This last axiom would allow, beyond properties such as the consideration that if a proposition $\phi$ is true independently of state then it would be impossible for any agent $i$ to know the negation of $\phi$ to be true without any logical contradiction; and the possibility of recovering knowledge from introspective knowledge, the certainty of one agent obtaining knowledge from knowing another agent's knowledge, that is, given a agent $i$ that knows that agent $j$ knows $\phi$, we would have: $  s \Vdash \mathscr{K}_i\mathscr{K}_j \phi \implies  s \Vdash \mathscr{K}_i \phi$.

It is from this last axiom that the naïve treatments of the Frauchinger-Renner Gedankenexperiment arrive at a logical contradiction. Since in quantum mechanics measurements almost always alter the state of the system being measured it would not make sense to consider that the original fifth knowledge axiom would make any sense, \textit{e.g.}, if for a given electron the assertion: $q \equiv$ ``the projection on the $z$ direction of the spin of the electron is in the positive $z$ direction"$\,$is to be considered in a situation in which an experimenter $A$ has realized the measurement.

If we resignify the meaning of $\mathscr{K}_A\,q$ to have the semantic meaning of ``$A$ has realized a measurement obtaining as result, or otherwise acquired the knowledge that $q$", then while it would make sense to write $\mathscr{K}_A\,q$, axiom $5$ would make this imply that assertion $q$ by itself would be true independently of any particular measurement experiment. However this makes no sense since, if another experimenter $B$ decided to measure the spin projection in the $x$ direction, then this would perturb the electron and make uncertain any subsequent measurement in the $z$ direction, making it possible not to be anymore in the positive $z$ direction as $q$ would imply.

In \cite{Nur} this last fifth axiom is first refined to a trust axiom (perhaps a better name would have been amenable information axiom) where the existence of a ``trust" relation $i \trusts j$ between agents $i$ and $j$ could be used to filter the cases in which the logic expression reduction $s \Vdash \mathscr{K}_i\mathscr{K}_j \phi \implies s \Vdash \mathscr{K}_i \phi\,,\, \forall\; s \in \Sigma,\phi \in \Phi$ would be allowed. Understanding $A \trusts B$ to mean a trust that if $A$ were to repeat the same measurement as $B$ did, then $A$ would always get the same result as $B$, in other words $A$ trusts that $B$ has not made a destructive measurement, or failed to keep the system isolated, or immediately after taking that measurement has made, or allowed another person to make, a incompatible measurement that would perturb the system in such a way as to make his measurement not repeatable.

From this considerations one can clearly see that the Trust axiom $A \trusts B \iff \mathscr{K}_A \mathscr{K}_B \,q  \Rightarrow \mathscr{K}_A \,q$ makes sense quantum mechanically, since as $A$ could in principle repeat the same measurement as $B$ and get the same result, then it makes as much sense for $A$ to know $q$ as does for $B$.

Nevertheless, this still leads to a contradiction if only this criterion of compatible measurements is used.

\section{Frauchinger-Renner No-Go Result} 

The inconsistency result of Frauchinger and Renner \cite{FrauRen} can be expressed as a theorem:

\begin{theorem}[Frauchinger-Renner]
	Any theory that satisfies the assumptions given by \textbf{Born's rule (Q)}, the ability to \textbf{generalize an assertion (C)} about the information available to an agent, with the possibility of simplifying information without any obstructions related to the receptacle of this information; and the \textbf{univaluation (S)} of measurement results of any one system, these assumptions lead to a contradiction in the scenario described by the Gedankenexperiment.
\end{theorem}

To clarify what the conditions \textbf{(Q), (C)} and \textbf{(S)} mean we will elaborate on them in what follows:

-- Condition  \textbf{(Q)} is plain old Born's rule, which for all intents and purposes here just means that for a normalized vector $\psi \in \mathscr{H}$, if $P_\xi \, \psi = \psi$ then $\xi = \psi$ and therefore $\langle \psi,\,P_\xi \, \psi \rangle =1$.

-- To what pertains condition \textbf{(C)}, the ability to generalize an assertion about the available information we must consider a type of metalogic of the system to describe logical relations between the agents themselves. For completeness we first consider this condition to mean the trust axiom, with the trust being given by the criterion of compatible measurements.

-- Lastly, condition \textbf{(S)} means that a agent cannot measure two different results for the same value at the same time, that is, if $P_\xi \, \psi = \psi$ at a time $t$, then $(\id - P_\xi)\psi = \mathbf{0}$ also at time $t$. Equivalently in the notation of Kripke semantics, if $\mathscr{K}_A \bigr[(P_\xi \, \psi \neq 0) \wedge (P_\xi \, \psi = 0) \bigl]$ then 
\begin{equation}
	\mathscr{K}_A (P_\xi \, \psi \neq 0) \wedge \left(\neg \, \mathscr{K}_A (P_\xi \, \psi \neq 0)\right),
\end{equation} 
 which is a usual, global, contradiction $A \wedge \neg A$, whereas the $\mathscr{K}_B[ A \wedge \neg A]$ could be seen as a local, \textit{i.e.}, to agent $B$, contradiction.   

Having made all of these important semantic considerations, we can go on to prove the Frauchinger-Renner no-go result.

\begin{proof}
	
	It is considered, just as in the Wigner's friend Gedankenexperiment, that each experimenter $F_1$ and $F_2$ can have their states measured to determine the measured outcome of the two-level system that they perceived, as a simplifying assumption we shall consider this to be representable with the orthogonal basis $\{\xi, \zeta\} \in \mathscr{H}_P$, where we have chosen other letters to better distinguish the states of the system and the states of the experimenters, of course one can criticize this simplification as non-representative of the complexity that is a rational agent, in which case the same experiment can be carried out mechanically with a second two-level quantum system being used as a replacement for the state of the experimenter memory where the measurement result is held.
	
	As occurs in QM , each experimenter has to consider the evolution of all other quantum systems, that he is not measuring, as unitary (all constituent parts of the experimental setup are fundamentally quantum systems and even though each of the experimenters know that the other experimenters are going to perform measurements, he can only add the effects of these measurements is with correlations on the states implemented by isometries, since the state of that experimenter himself isn't getting entangled with the system being measured by another agent.)
	
	We can then consider which agents receive relevant information from which other agents, or in other words what are the trust relations between the measuring agents. Clearly since $W_1$ announces his sole measurement result for all the other agents at time $t=2$, then we must have that all other agents trust $W_1$ at $t=2$, in particular $W_2^{t=3} \trusts W_1^{t\geq2} $. Analogously, agent $W_2$ announces his sole measurement result for all the other agents at time $t=3$, then we must have that all other agents trust $W_2$ at $t=3$, in particular $F_1^{t=4} \trusts W_2^{t=3}$.
	
	If $W_1 $ can indirectly find out the information of the measurement of $F_2$ at time $t=1$, then $W_1$ can trust this information since the quantum system $S_2$ that $F_2$ is measuring is and remains after its measurement isolated form $W_1$, that is  $W_1^{t\geq2} \trusts F_2^{t=2}$.
	
	And finally since the measurements that $F_2$ and $F_1$ make are compatible then any information one gets about the other they can trust as long as those systems have not been disturbed at those times by someone else, in particular $F_2^{t=2} \trusts F_1^{t=0,1,2}$.
	
	Synthesizing these particular relations into a trust hierarchy one gets:
	\begin{equation}\label{hie}
		F_1^{t=4} \trusts W_2^{t=3} \trusts W_1^{t\geq 2} \trusts F_2^{t=2} \trusts F_1^{t=0,1,2}.
	\end{equation}
	
	From the outset the following assertions are true because of the way the Gedankenexperiment was arranged and because every agent knows the arrangement of the experiment:
	
	If $F_1^{t=1}$ measures $\psi$ from system $S_1$, then $F_1$ knows that after $F_2$ measures $S_2=\sqrt{\frac{1}{2}} (\varphi+\psi)$ in the basis $\{\varphi,\psi\}$ it has as much chance of measuring $\varphi$ as it has for measuring $\psi$ therefore $F_1$ arrives at the conclusion that $W_2$ will be measuring the lab $L_2$ in the state $\sqrt{\frac{1}{2}}(\varphi \otimes \xi + \psi \otimes \zeta)$, which is orthogonal to $\chi = \sqrt{\frac{1}{2}}(\varphi \otimes \xi - \psi \otimes \zeta )$, hence $P_\chi L_2 =0$. In symbols:
	\begin{equation}\label{prop1}
		\Vdash \mathscr{K}_{F_1^{t <3}}\bigr[\mathscr{K}_{F_1^{t=1}}(S_1=\psi) \Rightarrow \mathscr{K}_{W_2^{t=4}}(P_{\chi}L_2=0)\bigr]
	\end{equation}
	
	Since $F_2$ knows from the outset of the experiment that $F_1$ will send him either a $S_2=\varphi$ state if $F_1$ measures $\varphi$, or a $S_2=\sqrt{\frac{1}{2}} (\varphi+\psi)$ if $F_1$ measures $\psi$, then if $F_2$ measures $S_2 =\psi$ he will know that $F_1$ had measured $S_1 =\psi$. Or with the formula:
	\begin{equation}\label{prop2}
		\Vdash \mathscr{K}_{F_2^{t <3}}\bigr[\mathscr{K}_{F_2^{t=1}}(S_2=\psi) \Rightarrow \mathscr{K}_{F_1^{t=0,1,2}}(S_1=\psi)\bigr]
	\end{equation}
	
	Since $W_1$ knows that when he is going to make a measurement of the lab $L_1$, he considers that the global state is at that time equal to 
	\begin{equation}
	\frac{1}{\sqrt{3}}(\varphi_{S_1} \otimes \xi_{W_1} \otimes \varphi_{S_2} \otimes \xi_{W_2} +\psi_{S_1} \otimes \zeta_{W_1} \otimes \varphi_{s_2} \otimes \xi_{W_2} +\psi_{S_1} \otimes \zeta_{W_1} \otimes \psi_{S_2} \otimes \zeta_{W_2})
    \end{equation}
    (where we have put subscripts to facilitate understanding of whose state that vector is associated with) hence he can factor this state in the basis that he is measuring with to get
    \begin{equation}
    \sqrt{\frac{2}{3}}\frac{1}{\sqrt{2}}(\varphi_{S_1} \otimes \xi_{W_1} + \psi_{S_1} \otimes \zeta_{W_1} ) \otimes \varphi_{S_2} \otimes \xi_{W_2} + \frac{1}{\sqrt{3}}\psi_{S_1} \otimes \zeta_{W_1} \otimes \psi_{S_2} \otimes \zeta_{W_2},
    \end{equation}since the first part of the global state is perpendicular to the vector $\chi \equiv \sqrt{\frac{1}{2}} (\varphi_{S_1} \otimes \xi_{W_1} - \psi_{S_1} \otimes \zeta_{W_1})$ associated with the projector $P_{\chi}$ that $W_1$ uses to make his measurement, then if he gets $P_{\chi} L_1 \neq 0$ that means that the only part of the global state that overlaps with this observation is the part $\frac{1}{\sqrt{3}}\psi_{S_1} \otimes \zeta_{W_1} \otimes \psi_{S_2} \otimes \zeta_{W_2}$, which makes $W_1$ certain that $F_2$ measured $\psi$ at time $t=2$, that is:  
	\begin{equation}\label{prop8}
		\Vdash \mathscr{K}_{W_1^{t <4}}\bigr[\mathscr{K}_{W_1^{t=3}}(P_{\chi} L_1 \neq 0) \Rightarrow \mathscr{K}_{F_2^{t=1,2}}(S_2=\psi)\bigr]
	\end{equation}
	
	Since before the beginning of the experiment the participants must communicate amongst themselves to set up the experiment in the first place, then they can also inform each other of the implications that each deduced based on a possible measurement outcome, in particular the previous tree assertions can be made know to $F_2,\,W_1$ and $W_2$ respectively:      
	\begin{equation}\label{prop3}
		\Vdash \mathscr{K}_{F_2^{t =2}}\mathscr{K}_{F_1^{t <3}}\bigr[\mathscr{K}_{F_1^{t=1}}(S_1=\psi) \Rightarrow \mathscr{K}_{W_2^{t=4}}(P_{\chi}L_2=0)\bigr]
	\end{equation}
	\begin{equation}\label{prop9}
		\Vdash \mathscr{K}_{W_1^{t =3}}\mathscr{K}_{F_2^{t <3}}\bigr[\mathscr{K}_{F_2^{t=1}}(S_2=\psi) \Rightarrow \mathscr{K}_{F_1^{t=0,1,2}}(S_1=\psi)\bigr]
	\end{equation}
	\begin{equation}
		\Vdash \mathscr{K}_{W_2^{t=4}}\mathscr{K}_{W_1^{t <4}}\bigr[\mathscr{K}_{W_1^{t=3}}(P_{\chi} L_1 \neq 0) \Rightarrow \mathscr{K}_{F_2^{t=1,2}}(S_2=\psi)\bigr]
	\end{equation}

	These in part can be communicated again to other agents and so on, up to third- and fourth-order statements like:
	\begin{equation}\label{prop10}
		\Vdash \mathscr{K}_{W_1^{t =3}}\mathscr{K}_{F_2^{t =2}}\mathscr{K}_{F_1^{t <3}}\bigr[\mathscr{K}_{F_1^{t=1}}(S_1=\psi) \Rightarrow \mathscr{K}_{W_2^{t=4}}(P_{\chi}L_2=0)\bigr]
	\end{equation}
	\begin{equation}
		\Vdash \mathscr{K}_{W_2^{t =4}}\mathscr{K}_{W_1^{t =3}}\mathscr{K}_{F_2^{t =2}}\mathscr{K}_{F_1^{t <3}}\bigr[\mathscr{K}_{F_1^{t=1}}(S_1=\psi) \Rightarrow \mathscr{K}_{W_2^{t=4}}(P_{\chi}L_2=0)\bigr]
	\end{equation}
	\begin{equation}
		\Vdash \mathscr{K}_{W_2^{t =4}}\mathscr{K}_{W_1^{t =3}}\mathscr{K}_{F_2^{t <3}}\bigr[\mathscr{K}_{F_2^{t=1}}(S_2=\psi) \Rightarrow \mathscr{K}_{F_1^{t=0,1,2}}(S_1=\psi)\bigr]
	\end{equation}
	
	As we can only reduce the statements made about trusted agents, we shall begin by considering all of the previous statements that are within $W_2$'s knowledge,
	\begin{equation}\label{prop5}
		\Vdash \mathscr{K}_{W_2^{t=4}}\mathscr{K}_{W_1^{t <4}}\bigr[\mathscr{K}_{W_1^{t=3}}(P_{{\chi}} L_1 \neq 0) \Rightarrow \mathscr{K}_{F_2^{t=1,2}}(S_2=\psi)\bigr]
	\end{equation}
	\begin{equation}\label{prop6}
		\Vdash \mathscr{K}_{W_2^{t =4}}\mathscr{K}_{W_1^{t =3}}\mathscr{K}_{F_2^{t <3}}\bigr[\mathscr{K}_{F_2^{t=1}}(S_2=\psi) \Rightarrow \mathscr{K}_{F_1^{t=0,1,2}}(S_1=\psi)\bigr]
	\end{equation}
	\begin{equation}\label{prop7}
		\Vdash \mathscr{K}_{W_2^{t =4}}\mathscr{K}_{W_1^{t =3}}\mathscr{K}_{F_2^{t =2}}\mathscr{K}_{F_1^{t <3}}\bigr[\mathscr{K}_{F_1^{t=1}}(S_1=\psi) \Rightarrow \mathscr{K}_{W_2^{t=4}}(P_{\chi}L_2=0)\bigr]
	\end{equation}
	
	First we combine the bottom two statements:
	
	\begin{multline}
		\Vdash \mathscr{K}_{W_2^{t =4}}\mathscr{K}_{W_1^{t =3}}\mathscr{K}_{F_2^{t =2}}\biggr[\bigr[\mathscr{K}_{F_2^{t=1}}(S_2=\psi) \Rightarrow \mathscr{K}_{F_1^{t=0,1,2}}(S_1=\psi)\bigr] \wedge \\ \wedge \Bigr[ \mathscr{K}_{F_1^{t <3}}\bigr[\mathscr{K}_{F_1^{t=1}}(S_1=\psi) \Rightarrow \mathscr{K}_{W_2^{t=4}}(P_{\chi}L_2=0)\bigr] \Bigr]\biggr],
	\end{multline}
	and now we apply the trust condition $F_2^{t=2} \trusts F_1^{t=0,1,2} $ :
	\begin{align}
		\Vdash\mathscr{K}_{W_2^{t =4}}\mathscr{K}_{W_1^{t =3}}\mathscr{K}_{F_2^{t =2}}\Bigr[&\bigr[\mathscr{K}_{F_2^{t=1}}(S_2=\psi) \Rightarrow \mathscr{K}_{F_1^{t=0,1,2}}(S_1=\psi)\bigr] \wedge \nonumber \\ \wedge &\bigr[\mathscr{K}_{F_1^{t=1}}(S_1=\psi) \Rightarrow \mathscr{K}_{W_2^{t=4}}(P_{\chi}L_2=0)\bigr] \Bigr].
	\end{align}
	Using now the distribution axiom \ref{distr} for $F_2$:
	\begin{equation}
		\Vdash\mathscr{K}_{W_2^{t =4}}\mathscr{K}_{W_1^{t =3}}\mathscr{K}_{F_2^{t =2}}\bigr[\mathscr{K}_{F_2^{t=1}}(S_2=\psi) \Rightarrow \mathscr{K}_{W_2^{t=4}}(P_{\chi}L_2=0)\bigr].
	\end{equation}
	Now we combine this with the remaining statement about $W_2$'s knowledge:
	\begin{multline}\label{firstoner}
		\Vdash\mathscr{K}_{W_2^{t =4}}\mathscr{K}_{W_1^{t =3}} \biggl[ \bigl[\mathscr{K}_{W_1^{t =3}}(P_{\chi} L_1 \neq 0) \Rightarrow \mathscr{K}_{F_2^{t=1,2}}(S_2=\psi) \bigr] \wedge \\ \wedge \Bigl[\mathscr{K}_{F_2^{t =2}}\bigr[\mathscr{K}_{F_2^{t=1}}(S_2=\psi) \Rightarrow \mathscr{K}_{W_2^{t=4}}(P_{\chi}L_2=0)\bigr] \Bigr] \biggr].
	\end{multline}
	Again, we apply the trust condition $W_1^{t\geq 2} \trusts F_2^{t=2}$, followed by the distribution axiom for $W_1$, obtaining:
	\begin{equation}\label{krabre}
		\Vdash \mathscr{K}_{W_2^{t =4}}\mathscr{K}_{W_1^{t =3}}\bigr[  \mathscr{K}_{W_1^{t =3}}(P_{\chi} L_1 \neq 0) \Rightarrow \mathscr{K}_{W_2^{t=4}}(P_{\chi}L_2=0)\bigr].
	\end{equation}
	Note that the above holds for all possible outcome scenarios, and that $W_1$ and $W_2$   reach the conclusion $\mathscr{K}_{W_1^{t =3}}(P_{\chi} L_1 \neq 0) \Rightarrow \mathscr{K}_{W_2^{t=4}}(P_{\chi}L_2=0)$ before the experiment even starts. Additionally, at the end of the experiment, after time $t=3$, $W_1$ and $W_2$ learn of each other’s outcomes, in particular:
	\begin{equation}
		\forall \, s \in \Sigma: \,\,\,\,\,\, \bigl[s \Vdash  \mathscr{K}_{W_1^{t=2}} (P_{\chi} L_1 \neq 0) \bigr] \implies  \bigl[s \Vdash  \mathscr{K}_{W_2^{t=3}} \mathscr{K}_{W_1^{t=2,3}} (P_{\chi} L_1 \neq 0) \bigr].
	\end{equation}

	Now we run the experiment. As in the Hardy's paradox Gedankenexperiment, we analyze a specific scenario within all possible configurations of outcomes that exist in the experiment $\Sigma$, namely, we follow the scenario $s$ in which both $W_1$ reports ``$P_{\chi} L_1 \neq 0$"$\,$ and $W_2$ reports ``$P_{\chi} L_2 \neq 0$", by considering from the very beginning the global state of the experiment:
	$$\frac{1}{\sqrt{3}}(\varphi_{S_1} \otimes \xi_{W_1} \otimes \varphi_{S_2} \otimes \xi_{W_2} +\psi_{S_1} \otimes \zeta_{W_1} \otimes \varphi_{S_2} \otimes \xi_{W_2} +\psi_{S_1} \otimes \zeta_{W_1} \otimes \psi_{S_2} \otimes \zeta_{W_2})=$$
	\begin{multline}=\frac{1}{\sqrt{3}} \left[\left(\varphi_{S_1} \otimes \xi_{W_1} +\psi_{S_1} \otimes \zeta_{W_1}\right)\left(\varphi_{S_2} \otimes \xi_{W_2}+\frac{1}{2}\psi_{S_2} \otimes \zeta_{W_2}\right)\right]- \\ - \frac{1}{4\sqrt{3}}\left(\varphi_{S_1} \otimes \xi_{W_1} -\psi_{S_1} \otimes \zeta_{W_1}\right)\left(\varphi_{S_2} \otimes \xi_{W_2}+\psi_{S_2} \otimes \zeta_{W_2}\right) + \\ +\frac{1}{2\sqrt{3}}\left[\left(\sqrt{\frac{1}{2}}\left(\varphi_{S_1} \otimes \xi_{W_1} -\psi_{S_1} \otimes \zeta_{W_1}\right)\right)\left(\sqrt{\frac{1}{2}}\left(\varphi_{S_2} \otimes \xi_{W_2}-\psi_{S_2} \otimes \zeta_{W_2}\right)\right)\right]\end{multline}
	therefore the chance of this happening is $\left|\frac{1}{2\sqrt{3}}\right|^2=\frac{1}{12}$, hence: 
	\begin{equation}
		\exists\, s \in \Sigma: \,\,\,\,\,\, s \Vdash  \mathscr{K}_{W_1^{t=2,3}}(P_{\chi} L_1 \neq 0) \wedge \mathscr{K}_{W_2^{t=3}}(P_{\chi} L_2 \neq 0) 
	\end{equation}
	Given that $W_1$ communicates his result to all other agents in the experiment, we must also have that $W_2$ knows the result of $W_1$ at time $t=3$: 
	\begin{equation}
		\exists\, s \in \Sigma: \,\,\,\,\,\, s \Vdash  \mathscr{K}_{W_2^{t=3}} \left[\mathscr{K}_{W_1^{t=2,3}}(P_{\chi} L_1 \neq 0) \wedge \mathscr{K}_{W_2^{t=3}}(P_{\chi} L_2 \neq 0)  \right].
	\end{equation}
	Now we can combine this with the assertion (\ref{krabre})  to get:
	
	\begin{multline}
		\exists\, s \in \Sigma: \,\,\,\,\,\, s \Vdash  \mathscr{K}_{W_2^{t=3}} \Bigr[ \bigl[ \mathscr{K}_{W_1^{t=2}} (P_{\chi} L_1 \neq 0) \wedge \mathscr{K}_{W_2^{t=3}}(P_{\chi} L_2 \neq 0)\bigr] \wedge \\ \wedge  \mathscr{K}_{W_1^{t =3}}\bigr[  \mathscr{K}_{W_1^{t =2}}(P_{\chi} L_1 \neq 0) \Rightarrow \mathscr{K}_{W_2^{t=4}}(P_{\chi}L_2=0)\bigr] \Bigl]
	\end{multline}
	applying the trust relation $W_2^{t=3} \trusts W_1^{t=2} $:
	
	\begin{align*}
		\exists\, s \in \Sigma: \,\,\,\,\,\,\,\,\, s \Vdash  \mathscr{K}_{W_2^{t=3}} \Bigr[ &\bigl[ \mathscr{K}_{W_2^{t=3}}(P_{\chi} L_2 \neq 0) \wedge \mathscr{K}_{W_1^{t=2}} (P_{\chi} L_1 \neq 0) \bigr] \wedge \\ \wedge  &\bigr[  \mathscr{K}_{W_1^{t =2}}(P_{\chi} L_1 \neq 0) \Rightarrow \mathscr{K}_{W_2^{t=4}}(P_{\chi}L_2=0)\bigr] \Bigl]
	\end{align*}
	
	using now the distribution axiom for $W_1$:
	
	\begin{equation}
		\exists\, s \in \Sigma: \,\,\,\,\,\,\,\,\, s \Vdash  \mathscr{K}_{W_2^{t=3}} \bigl[ \mathscr{K}_{W_2^{t=3}}(P_{\chi} L_2 \neq 0) \wedge \mathscr{K}_{W_2^{t=4}}(P_{\chi}L_2=0) \bigr].
	\end{equation}
	
	Finally, we can apply the trivial trust relation $W_2^{t=3} \trusts W_2^{t=3}$ to obtain:
	\begin{equation}
		\exists\, s \in \Sigma: \,\,\,\,\,\, s \Vdash  \mathscr{K}_{W_2^{t=3}} \bigl[ (P_{\chi} L_2 \neq 0) \wedge (P_{\chi}L_2=0) \bigr],
	\end{equation}
	we now use condition \textbf{(S)} to get a global contradiction from a local one:
	\begin{equation}
		\left(\mathscr{K}_{W_2^{t=3}} \wedge \neg \mathscr{K}_{W_2^{t=3}}\right)\left[P_{\chi}L_2=0\right].
	\end{equation}
	
	As a matter of fact this contradiction can be achieved by all agents at he end of the experiment by considering the emulation of the thought process of $W_2$ and his final announcement at $t=3$.
	
	That is, in scenario $s \in \Sigma$, $F_1$ knows that he measured $S_1=\psi$ at $t=0$ and by proposition (\ref{prop1})
	\begin{equation}
		\left[s \Vdash \mathscr{K}_{F_1^{t=1}}(S_1=\psi) \right] \wedge \left[s \Vdash \mathscr{K}_{F_1^{t<3}}\left[\mathscr{K}_{F_1^{t=1}} (S_1 = \psi) \Rightarrow \mathscr{K}_{W_2^{t=4}}(P_\chi L_2 = 0)\right] \right],
	\end{equation}
	and by the distribution axiom (\ref{distr}) (and the introspection axiom \ref{posit})
	\begin{equation}\label{prop21}
		s \Vdash \mathscr{K}_{F_1^{t<3}}\mathscr{K}_{W_2^{t=4}}(P_\chi L_2 = 0).
	\end{equation}
	
	$F_2$, in scenario $s \in \Sigma$, knows that he measured $S_2=\psi$ at $t=1$ and from propositions (\ref{prop2}), (\ref{prop3}) and the distribution axiom (\ref{distr}) (and again the introspection axiom \ref{posit})
	
	\begin{equation}
		s \Vdash \mathscr{K}_{F_2^{t=2}}\mathscr{K}_{F_1^{t<3}} \left[\mathscr{K}_{W_2^{t=4}}(P_\chi L_2 = 0)\right]
	\end{equation}
	using then the trust hierarchy (\ref{hie}) and the trust axiom
	\begin{equation}\label{prop22}
		s \Vdash \mathscr{K}_{F_2^{t=2}}\mathscr{K}_{W_2^{t=4}}(P_\chi L_2 = 0).
	\end{equation}
	
	At last, using completely analogous steps to the deduction of proposition (\ref{krabre}), where instead of propositions (\ref{prop5}), ((\ref{prop6})) and (\ref{prop7}) we use propositions (\ref{prop8}), (\ref{prop9}) and (\ref{prop10}), we obtain
	
	\begin{equation}
		\Vdash \mathscr{K}_{W_1^{t =3}}\bigr[  \mathscr{K}_{W_1^{t =3}}(P_{\chi} L_1 \neq 0) \Rightarrow \mathscr{K}_{W_2^{t=4}}(P_{\chi}L_2=0)\bigr],
	\end{equation}
	and considering that in scenario $s \in \Sigma, \mathscr{K}_{W_1^{t =3}}(P_{\chi} L_1 \neq 0)$, using the distribution axiom (\ref{distr}) (and yet again the introspection axiom \ref{posit})
	
	\begin{equation}\label{prop23}
		s \Vdash \mathscr{K}_{W_1^{t =3}}\mathscr{K}_{W_2^{t=4}}(P_{\chi}L_2=0).
	\end{equation}
	
	We can then re-obtain the global logical contradiction from each one of the agents, using propositions (\ref{prop21}), (\ref{prop22}) and (\ref{prop23}), the introspection axiom on those at time $t=4$ together with the trivial trust of each agent at latter times in themselves at earlier times, the information that in scenario $s$: $\mathscr{K}_{W_2^{t=4}}(P_{\chi}L_2\not=0)$, which $W_2$ informs to every other agent at the end of step $t=3$ and condition \textbf{(S)} to, again, get a global contradiction from a local one.
	\end{proof}
	
	 One can notice a striking similarity between this process of contradictory two-valuation of an assertion in this Gedankenexperiment and the also contradictory two-valuation of the Kochen-Specker theorem \cite{KochenSpecker}, even though no mention was present in the original paper.
	
	Although \cite{Nur} also used Kripke semantics \cite{ModalL} to simplify and organize the original argument and also considered contextuality as a solution to the problem, the notion of contextuality that was considered is the one called ``two-dimensional semantics", which differs from the more common definition, where a context is any given abelian von Neumann sub-algebra of the von Neumann algebra of observables \cite{Doring}.
	
	In two-dimensional semantics the entire list of observers that have made measurements and provided information is appended to the epistemic operators $\mathscr{K}$ and at each measurement a new observer and its measurement is appended to the list. In this sense, this list becomes the ``context" of under which measurements those assertions were made. 
	
	The usual notion of contextuality simplifies this description by only looking at the orthogonal basis under which any of the measurements were performed, hence we can change the trust axiom to reflect this constraint:
	
	\begin{definition}[Contextual information axiom]
		We say that an agent $i$ is amenable to the information of $j$, and denote it by $i \trusts j$, if and only if the context $\mathscr{C}$ of $i$ is the same as the context of $j$ and:
		\begin{equation}
			s \Vdash \mathscr{K}_i\hspace{-1.1pt}\rest{\mathscr{C}}\;\mathscr{K}_j\hspace{-1.1pt}\rest{\mathscr{C}}\; \phi \implies s \Vdash \mathscr{K}_i\hspace{-1.1pt}\rest{\mathscr{C}}\; \phi\,,\,\,\,\,\, \forall\; s \in \Sigma,\phi \in \Phi
		\end{equation}
	\end{definition} 
	
	And we also have to update the distribution axiom as to account for the possibility of different contexts.
	
	\begin{definition}[Contextual distribution axiom]
	 For the same context $\mathscr{C}$ of $i$ :
		\begin{equation}
			s \Vdash \left(\mathscr{K}_i\hspace{-1.1pt}\rest{\mathscr{C}}\; \phi\right) \wedge \left(\mathscr{K}_i\hspace{-1.1pt}\rest{\mathscr{C}}\; (\phi \Rightarrow \psi)\right) \iff s \Vdash \mathscr{K}_i\hspace{-1.1pt}\rest{\mathscr{C}}\; \left(\phi \wedge (\phi \Rightarrow \psi)\right) \implies s \Vdash \mathscr{K}_i\hspace{-1.1pt}\rest{\mathscr{C}}\; \psi\,,\,\,\,\,\, \forall\; s \in \Sigma,\phi \in \Phi
		\end{equation}
	\end{definition} 
	
	Using this new axiom it becomes clear that $W_1^{t\geq 2} \not \hspace{-5pt} \trusts F_2^{t=2}$ and $F_1^{t=4} \not \hspace{-5pt} \trusts W_2^{t=3}$, since the contexts in which each of those makes their measurements is different, hence one cannot arrive at the contradiction since the assertion (\ref{firstoner}) now given by 
	\begin{multline}
	\Vdash\mathscr{K}_{W_2^{t =4}}\hspace{-1.1pt}\rest{\mathscr{C}_2}\;\mathscr{K}_{W_1^{t =3}}\hspace{-1.1pt}\rest{\mathscr{C}_2}\;\; \biggl[ \bigl[\mathscr{K}_{W_1^{t =3}}\hspace{-1.1pt}\rest{\mathscr{C}_2}\;(P_{\chi} L_1 \neq 0) \Rightarrow \mathscr{K}_{F_2^{t=1,2}}\hspace{-1.1pt}\rest{\mathscr{C}_1}\;(S_2=\psi) \bigr] \wedge \\
	\wedge \Bigl[\mathscr{K}_{F_2^{t =2}}\hspace{-1.1pt}\rest{\mathscr{C}_1}\;\,\bigr[\mathscr{K}_{F_2^{t=1}}\hspace{-1.1pt}\rest{\mathscr{C}_1}\;(S_2=\psi) \Rightarrow \mathscr{K}_{W_2^{t=4}}\hspace{-1.1pt}\rest{\mathscr{C}_2}\;(P_{\chi}L_2=0)\bigr] \Bigr] \biggr]
    \end{multline} 
    cannot be reduced anymore to an equivalent of (\ref{krabre}), since it isn't possible to combine the distribution axiom with the trust axiom to simplify that conjunction, \textit{i.e.}, before we could use that $\mathscr{K}_a(\mathscr{K}_a \phi \Rightarrow \mathscr{K}_b \psi) = \mathscr{K}_a\mathscr{K}_a \phi \Rightarrow \mathscr{K}_a\mathscr{K}_b \psi$ and the trust $a \trusts b$ to simplify this to $\mathscr{K}_a \phi \Rightarrow \mathscr{K}_a \psi$ and then to $\mathscr{K}_a(\phi \Rightarrow \psi)$ and then factor the common epistemic operators to simplify chains of implications by the classical hypothetical syllogism. Without the $\trusts$ relation between the epistemic operators in the implication, this is not possible anymore.
	
	The impossibility of such reductions is implicitly assumed by \cite{Nur} to mean that there is a hard limitation on the reasoning of the agents that is imposed on them without their noticing which makes them perceive the information received by measurements in other contexts as independent from their own. A much simpler interpretation is that the agents \textit{do, \textbf{in fact}, know} quantum mechanics, hence they must also know the Kochen-Specker theorem \cite{KochenSpecker}, hence they know they cannot use the measurement results from another context in their own without excluding the possibility of arriving at contradictory results. 
	
	A more standard way to rephrase that is to say: For the agents to actually be considered to be using QM in their deductions, they have to carefully consider the compatibility, or lack thereof, of the measurement results when those are being logically combined.
	
	\vspace{-11pt}
	
	\section{Complications of generalizing to QFTs}

    The Gedankenexperiment proposed by Frauchinger and Renner relies ubiquitously on the existence of sharp states, and their associated projectors. Though it is a widely know result in the literature \cite{Driessler, Antoni, Yngvason}, that the algebras of local observables are von Neumann algebras of Type III.
    
    To better understand the meaning of this, we remember the reader that the relation of Murray-von Neumann equivalence $p \sim q$ (which means: $\exists v, \, v^*v = p$ and $vv^*=q$, \cite{Takesaki}) is logically independent of the projector ordering relation $u < w$. Hence, the fact that every non zero projector in a Type III von Neumann algebra is infinite, i.e., $\forall p \in P(\mathfrak{A}), \, \exists q \in P(\mathfrak{A}), \, q < p$ and $q \sim p$; for $P(\mathfrak{A})$ equal to the set of projectors of algebra $\mathfrak{A}$.
    
    Hence we can never talk about sharp states, which would give rise to the yes-no propositions, in local algebras of observables. Which means we lose the basic building blocks of a Birkhoff-von Neumann quantum logic.
    
    Though, the Kochen-Specker theorem still holds \cite{Doring}. Which means that the basic contextuality property that the Frauchinger-Renner Gedankenexperiment suggests is still present in Type III von Neumann algebras, even if the property cannot be brought to a logical (and possibly paradoxical) statement in a Birkhoff-von Neumann framework.
	
	\section{Conclusion}
	The contradiction found by Frauchinger and Renner amounts to a implied supposition in a non-contextual logic existing above quantum mechanics, which directly violates the results of the Kochen-Specker theorem, making the contradiction be a result of the contradicting hypothesis rather than from quantum mechanics itself.
	
	Although, the largely algebraic treatment of this Gedankenexperiment, and of it's solution, could seem to amount to a straightforward generalization of this problem from non-relativistic quantum mechanics to quantum field theory by using the algebraic quantum field theory approach, the fact that the majority of the algebras, of physical interest, present in QFT, namely the ones made out of local observables, are von Neumann algebras of Type III, means that there would be no minimal projections in the lattice of projections of the von Neumann algebra of local observables \cite{QFTLogic}, and hence no sharp states would ever be available locally.
	
	Therefore the contradiction that in this Gedankenexperiment presented itself as a, state dependent, logical contradiction, at a quantum field theory setting could at most become a statistically improbable deviation of measurement statistics. 
	
	\section{ACKNOWLEDGMENTS}
	
	Felipe Dilho Alves is supported by CAPES, grant number 88887.977899/2024-00.

\bibliographystyle{plain} 
\bibliography{refs-Art1}           

\end{document}